\documentclass[8pt]{article}
\usepackage{graphicx}
\usepackage{pgfplots}
\usepackage{newlfont}
\usepackage{tikz}
\usepackage{enumitem}
\usepackage{rotating}
\usepackage{adjustbox}
\usepackage{overpic}
\usepackage{float}

\usepackage[colorlinks=true, linkcolor=blue, urlcolor=blue, citecolor=blue]{hyperref}
\usepackage{amsmath,amsthm,amssymb}
\usepackage{subfigure}

\usepackage{graphicx}
\usepackage{amsmath, amsfonts}
\usepackage{dsfont}
\usepackage{booktabs}
\usepackage{fourier}
\usepackage{xcolor}
\usepackage{array}
\usepackage{makecell}
\usepackage{algorithm}
\usepackage{algpseudocode}
\usepackage{lscape}
\usepackage{threeparttable}
\usepackage{hyperref}
\usepackage[T1]{fontenc}
\usepackage[tableposition=top]{caption}
\usepackage{tabularx}
\newtheorem{thm}{Theorem}[section]

\newtheorem{lem}[thm]{Lemma}

\newtheorem{defn}[thm]{Definition}

\newtheorem{example}[thm]{Example}
\numberwithin{equation}{section}
\usepackage{mathrsfs}

\begin{document}
	\title{\textbf{Image Restoration with Impulse Noise Based on Minimax Concave Penalty: Grayscale and Color Images}}
	\author{
		Reza Parvaz$^a$ \footnote{Corresponding author, Email:
			\href{rparvaz@uma.ac.ir}{rparvaz@uma.ac.ir},~
			\href{mailto:reza.parvaz@yahoo.com}{reza.parvaz@yahoo.com}}
	}
	\date{}
	\maketitle
	\begin{center}
		$^a$Department of Engineering Sciences, Faculty of Advanced Technologies, 
		University of Mohaghegh Ardabili, Ardabil, Iran
	\end{center}
	\begin{abstract}
		\noindent
In today's world-where social networks play a significant role-images serve as one of the most important tools 
for transferring concepts. Image degradation caused by noise or blurring hinders the effective transmission of 
the information we expect to derive from an image. Consequently, developing image restoration methods capable 
of revealing crucial information is of great importance. In this work, a model for impulse noise image restoring 
is proposed by using a Minimax Concave Penalty. The proposed model is applicable to both color and grayscale images. 
Then, a numerical algorithm based on the alternating direction method of multipliers is introduced, and the 
convergence of the method is analyzed both theoretically and numerically.
In the final step, the proposed method is evaluated through simulations on several examples, and the results 
demonstrate its effectiveness.

	\end{abstract}
	\vskip.3cm
	
	\section{Introduction and Problem Description}
Among human senses, vision is the most important factor in perceiving information. 
Consequently, images--which facilitate the recording and transmission of information--can 
be considered the most crucial data for comprehension.
However, various factors during the generation and transmission of images can degrade their quality. 
For instance, camera shake causes a phenomenon known as blurring, while noise is another factor that 
leads to the loss of image information.	
However, various factors during the generation and transmission of images can degrade their quality. 
For instance, camera shake causes a phenomenon known as blurring, while noise is another factor that 
leads to the loss of image information.	
Various types of noise affect images, depending on the nature of the image; for instance, Poisson noise is 
a type of noise that commonly affects astronomical and medical images.
Another type of noise that affects the image due to  erroneous transmission, faulty 
switching during the imaging process, and other factors,
is impulse noises \cite{nn01,nna01,nn02,nn03}.
The  salt--and--pepper noise and the random--valued noise \cite{nn01,nna01} can
be considered well-known examples of these types of noise.\\
 
Considering impulse noise, the degradation model for an image is written as 
\begin{align*}
f=Ku+\epsilon,
\end{align*}
where $u,f \in \mathbb{R}^{nm\times 1}$ are the vector forms of the clear 
image ($U  \in \mathbb{R}^{n\times m}$), and the observed degraded 
image  ($F  \in \mathbb{R}^{n\times m}$), respectively.
Also, $\epsilon \in \mathbb{R}^{nm\times 1}$ denotes  the impulse noise  
in the observed image.
In this  model, $K\in \mathbb{R}^{nm \times nm}$ represents a matrix 
obtained using the point spread function and the boundary conditions \cite{nn04}.
Directly solving these types of problems is not feasible due to their numerical 
ill-conditioning; therefore, the total variation (TV)  method is proposed for 
their solution \cite{nn05}. 
Although machine learning-based methods have been developed in recent years to 
solve these types of problems \cite{nn09,nn010}, 
issues such as the substantial data requirements of training and the need for 
significant computational resources
have prevented them from completely replacing classical methods \cite{nn011}.
In recent years, TV based methods have been widely employed to solve such problems
\cite{na1,na2}.
The $\ell_0$ norm as data fidelity in TV and the  $\ell_1$
norm on gradients for regularization is proposed in \cite{nn1}.
Another method based on $\ell_0$ is studied in \cite{nna1}.
 Despite its acceptable results, this model has certain flaws, including that the
$\ell_0$ term is nonconvex and discontinuous, and minimizing it exactly is NP-hard in general.
In \cite{nn2,nn3}, $\ell_1$ TV is introduced.
Although this model is convex, it is well known that it suffers from bias (e.g., shrinkage 
of large residuals and contrast loss).
Another method based on
$\ell_p$ TV ($0<p<1$) is studied in \cite{nn4}. The $\ell_p$  fidelity term is 
nonconvex and non-Lipschitz and the overall optimization problem is nonsmooth.
In the models studied above, the $\ell_1$ and $\ell_p$ norms are introduced as 
alternatives to the $\ell_0$ norm. The Minimax Concave Penalty, presented in \cite{nn08}, 
offers another approach for replacing and approximating the $\ell_0$ norm.
Unlike the 
$\ell_1$
penalty, which is biased, this penalty is continuous and unbiased while still promoting sparsity.
Although this alternative may not solve the problem of non-convexity,
non-convex penalties can yield a more accurate estimation of the underlying signal
\cite{nn06,nn07}.
This penalty has been used to solve the deblurring problem in \cite{nn5,nn6} and the Poisson 
noise removal problem in \cite{nn7}. 
 Also, an interesting discussion regarding the structure of this type of penalty is studied in \cite{nn8}.
 Based on the above discussion, in this paper, we propose a impulse noise removal model 
 based on this penalty. In the proposed model, 
the framelet transform and fractional order total varational (FTV) methods are employed 
to enhance sparsity and preserve edges.
Although the proposed model is presented for grayscale images, it is also being extended to color images.\\

The overall structure of the article can be summarized as follows.
Section \ref{sec1} provides a brief overview of the tools and fundamental concepts used in the article.
The proposed model for solving the impulse noise removal problem and its extension to color 
images are presented in Section \ref{sec2}. In addition, a numerical method 
based on the alternating direction method of multipliers 
is also provided in this section to solve the proposed model.
The convergence of the proposed method is studied in Section \ref{sec4}.
To evaluate the performance of the proposed method for image restoration, 
several numerical examples are presented and analyzed in Section \ref{sec5}.
\section{Preliminaries and Tools} \label{sec1}
In this section, we present a summary of the tools and concepts that are used in the subsequent sections.\\

\begin{defn} \cite{nn9,nn10}
For an image $u$, the fractional order total variation norm  
is defined as 
\begin{align*}
\|u\|_{\mathrm{FTV}}=\sum_{i,j}\sqrt{ |D^{\alpha}_x u_{i,j}|^2+|D^{\alpha}_y u_{i,j}|^2},
\end{align*}
where $D^{\alpha}_x$ and $D^{\alpha}_y$ denote de the G-L fractional-order 
of order $\alpha$ and defined as:
\begin{align*}
&D^{\alpha}_x u_{i,j}=\sum^{k-1}_{l=0} (-1)^l \frac{\Gamma(\alpha+1)}{\Gamma(l+1)\Gamma(\alpha-l+1)}u_{i-k,j},\\
&D^{\alpha}_y u_{i,j}=\sum^{k-1}_{l=0} (-1)^l \frac{\Gamma(\alpha+1)}{\Gamma(l+1)\Gamma(\alpha-l+1)}u_{i,j-k}.
\end{align*}	
Here $k$ is the number of neighboring pixels that are used to obtain the derivative.
\end{defn}

Another tool playing a fundamental role in this article is the Minimax Concave Penalty (MCP). 
We present a definition of this concept below. Readers can find further information on this topic in \cite{nn08,nn8}.

\begin{defn}
The MCP function is defined as follows \cite{nn08,nn8}
\begin{align*}
h_{\gamma, \eta}(x)=\int^{|x|}_0 \big(\gamma-\frac{|t|}{\eta}\big)_{+}dt,
\end{align*}
where $\gamma>0$ and $\eta>1$ are constants.
\end{defn}

\begin{defn}
For a vector $x \in \mathbb{R}^n$, the vector MCP  is defined by \cite{nn8}
\begin{align*}
\| x\|_\mathrm{MCP}=\sum_{i=1}^{n}h_{\gamma,\eta}(x_i).
\end{align*}
 \end{defn}

By considering the above definitions and following an argument similar to that in \cite{nn8},
we state the following lemma, which will be required later for solving the model.	
		
\begin{lem}\label{lemm1}
	Let $y\in \mathbb{R}^n$, also let $\alpha>0$, then the solution of the following optimization problem
\begin{align*}
x^{\ast}=\arg \min_{x} \frac{1}{2}\|x-y\|^2_2+\alpha \|x\|_\mathrm{MCP},	
\end{align*}
is obtained as
\begin{align*}
	x^{\ast}=\psi_{\alpha}^{\gamma,\eta}(y):= 
	\min \left\{ |y|, \max\left\{ \frac{\gamma (|y|-\alpha \eta)}{\gamma-\alpha},0 \right\} \right\}
	\odot \mathrm{sign}(y),
\end{align*}
where $\odot$ denotes the Hadamard product.
\end{lem}
The next tool employed in the presented model is the concept of the frame. Given the breadth of this concept, the reader can find sufficient information regarding its definition, as well as its computational and practical methods, in \cite{nn11}. In this paper, we utilize $W$ as the framelet transformmatrix constructed from the low-pass and high-pass filters
$
h_0=\frac{1}{4}[1,2,1], h_1=\frac{\sqrt{2}}{4}[1,0,-1]$ and $h_2=\frac{1}{4}[-1,2,-1]$.\\
Some of the symbols used in the text are as follows:
 $I, F$ and $\ast$ are used for identity matrix, the fast Fourier transform (FFT) and complex conjugate, respectively
	
	\section{Proposed Model and Numerical Algorithm} \label{sec2}
In this section, we introduce the model used for image restoration and present a numerical method to solve it; finally, the proposed algorithm is extended to color images.

	\subsection{Proposed MCP-Based Model }
	To solve the impulse noise problem, we introduce the following model:
 \begin{align}\label{mod1}
 	\min_u \|Ku-f\|_{\mathrm{MCP}}+\lambda_1\|Wu\|_{\mathrm{MCP}}+\lambda_2 \|u\|_{\mathrm{FTV}},
 \end{align}
where $\lambda_1$ and $\lambda_2$ are given positive regularization parameters.
The proposed model is nonconvex, and we employ the alternating direction method of multipliers to solve it. 
Although this method is originally developed for convex problems, recent studies indicate that it also yields 
convergent numerical algorithms when applied to nonconvex models \cite{nn12,nn14}.\\

In the first step of ADMM, by introduce auxiliary variables $\boldsymbol{m}=\{m_i\}^{3}_{i=1}$, the
 model \eqref{mod1} can be rewritten as
\begin{align}
&\nonumber \min_{u,\boldsymbol{m}} \|m_1\|_{\mathrm{MCP}}+\lambda_1\|m_2\|_{\mathrm{MCP}}+\lambda_2\|m_3\|_{\mathrm{FTV}},\\
&\text{s.t.}\quad m_1=Ku-f,\quad m_2=Wu, \quad m_3=u.\label{eq1}
\end{align}	
Now by defining 
$\boldsymbol{n}=\{n_i\}^{3}_{i=1}$ and $\boldsymbol{\mu}=\{\mu_i\}^{3}_{i=1}$ as Lagrange multipliers 
and penalty parameters, respectively, the augmented Lagrangian for \eqref{eq1} is obtained as:
\begin{align}
\nonumber L(u,\boldsymbol{m},\boldsymbol{n},\boldsymbol{\mu})=& 
\|m_1\|_{\mathrm{MCP}}+\lambda_1\|m_2\|_{\mathrm{MCP}}+\lambda_2\|m_3\|_{\mathrm{FTV}}\\
\nonumber &+\mu_1 \langle Ku-f-m_1,n_1\rangle+\frac{\mu_1}{2}\|Ku-f-m_1\|^2_2\\ 
\nonumber &+\mu_2 \langle Wu-m_2,n_2\rangle+\frac{\mu_2}{2}\|Wu-m_2\|^2_2\\ 
&+\mu_3 \langle u-m_3,n_3\rangle+\frac{\mu_3}{2}\|u-m_3\|^2_2. \label{eq3}
\end{align}	
Then the problem is $\min_{u,\boldsymbol{m} }L(u,\boldsymbol{m},\boldsymbol{n},\boldsymbol{\mu})$. 

\subsection{Solving the Subproblem}
In this subsection, we give the details of the ADMM iterative method. We assume that the values 
at the $j$ step are known and the objective is to calculate the solutions at the $j+1$ step.\\
	
\noindent By considering the augmented Lagrangian \eqref{eq3},  $u$ subproblem is obtained as 
\begin{align*}
u^{(j+1)}=&\arg\min_u \mu_1^{(j)} \langle Ku-f-m^{(j)}_1,n^{(j)}_1\rangle+\frac{\mu_1^{(j)}}{2}\|Ku-f-m^{(j)}_1\|^2_2\\
&+\mu_2^{(j)} \langle Wu-m^{(j)}_2,n^{(j)}_2\rangle+\frac{\mu^{(j)}_2}{2}\|Wu-m^{(j)}_2\|^2_2\\
&+\mu_3^{(j)} \langle u-m^{(j)}_3,n^{(j)}_3\rangle+\frac{\mu^{(j)}_3}{2}\|u-m^{(j)}_3\|^2_2.
\end{align*}	
The above subproblem can be rewritten as:
\begin{align}\label{eq2}
u^{(j+1)}=&\arg\min_u  \frac{\mu_1^{(j)}}{2} \|Ku-\xi_1\|^2_2+\frac{\mu^{(j)}_2}{2}\|Wu-\xi_2\|^2_2
+\frac{\mu^{(j)}_3}{2}\|u-\xi\|^2_2,
\end{align}
where:
\begin{align*}
\xi_1:=f+m^{(j)}_1-\frac{1}{\mu^{(j)}_1} n^{(j)}_1, \quad \xi_2:=m^{(j)}_2
-\frac{1}{\mu^{(j)}_2}n^{(j)}_2, \quad \xi_3:=m^{(j)}_3-\frac{1}{\mu^{(j)}_3}n^{(j)}_3.
\end{align*}
By using the optimal condition for \eqref{eq2}, the following equation is obtained:
\begin{align*}
	\big(\mu_1 K^TK+\mu^{(j)}_2W^TW+\mu^{(j)}_3\big)u^{(j+1)}=\mu^{(j)}_1 A^T\xi_1+\mu^{(j)}_2W^T\xi_2+\mu^{(j)}_3\xi_3.
\end{align*}
If we consider the boundary condition as a periodic condition, this equation can be written as 
\begin{align}\label{eq8}
	\big(\mu_1 \Lambda^{\ast}\Lambda+(\mu^{(j)}_2+\mu^{(j)}_3)I\big)Fu^{(j+1)}=\mu^{(j)}_1 
	\Lambda^{\ast}F\xi_1+\mu^{(j)}_2FW^T\xi_2+\mu^{(j)}_3F\xi_3.
\end{align}
Also,
$\Lambda$ denotes a diagonal matrix that entries are determined by the eigenvalues of the blurred matrix.
The solution of this equation is easily obtained using Fourier calculations.
\\

\noindent The subproblem $m_1$ is give as:
\begin{align*}
	m^{(j+1)}_1=\arg \min_{m_1} \|m_1\|_{\mathrm{MCP}}+\mu^{(j)}_1 \langle Ku^{(j+1)}
	-f-m_1,n^{(j)}_1\rangle+\frac{\mu^{(j)}_1}{2}\|Ku^{(j+1)}-f-m_1\|^2_2.
\end{align*}
By defining $\xi_4:=Ku^{(j+1)}-f+1/\mu^{(j)}_1n^{(j)}_1 $, the above problem can be written as:
\begin{align}\label{eq4}
	m^{(j+1)}_1=\arg \min_{m_1} \|m_1\|_{\mathrm{MCP}}+\frac{\mu^{(j)}_1}{2}\|m_1-\xi_4\|^2_2.
\end{align}
The solution of the problem \eqref{eq4} can be obtained using Lemma \ref{lemm1} as follows:
\begin{align}\label{eq9}
m_1=\psi_{ \frac{1}{\mu_1^{(j)}}  }^{\gamma,\eta}(\xi_4).
\end{align}

\noindent Subproblem $m_2$ is obtained as
\begin{align*}
	m^{(j+1)}_2=
	\arg \min_{m_2} \lambda_1\|m_2\|_{\mathrm{MCP}}+\mu_2 \langle Wu^{(j+1)}
	-m_2,n^{(j)}_2\rangle+\frac{\mu_2}{2}\|Wu^{(j+1)}-m_2\|^2_2.
\end{align*}
Considering  $\xi_5:=Wu^{(j+1)}+1/\mu^{(j)}_2n^{(j)}_2 $, the above problem can be written as:
\begin{align}\label{eq6}
	m^{(j+1)}_2=\arg \min_{m_2} \lambda_1 \|m_2\|_{\mathrm{MCP}}+\frac{\mu^{(j)}_2}{2}\|m_2-\xi_5\|^2_2.
\end{align}
By using  Lemma \ref{lemm1}, the solution of \eqref{eq6} is given as
\begin{align}\label{eq11}
m^{(j+1)}_2=\psi_{ \frac{\lambda_1}{\mu_2^{(j)}}  }^{\gamma,\eta}(\xi_5).
\end{align}

\noindent In the final stage, we consider the subproblem $m_3$ as
\begin{align}\label{eq7}
m^{(j+1)}_3=\arg \min_{m_3} \lambda_2\|m_3\|_{FTV}+\mu_3 \langle u^{(j+1)}-m_3,
n^{(j)}_3\rangle+\frac{\mu^{(j)}_3}{2}\|u^{(j+1)}-m_3\|^2_2.
\end{align}
By defining $\xi_6:=u^{(j+1)}+1/\mu^{(j)}_3n^{(j)}_3$, problem \eqref{eq7} is changed as 
\begin{align}\label{eq5}
m^{(j+1)}_3=\arg \min_{m_3} \lambda_2\|m_3\|_{FTV}+\frac{\mu^{(j)}_3}{2}\|m_3-\xi_6\|^2_2.
\end{align}
In order to solve this subproblem, the dual problem is used. Before obtaining the 
solution, we consider the following lemma.
\begin{lem}\label{lem2}
The dual problem of $J(u)=\min_u \|u\|_{\mathrm{FTV}}$ is obtained as \cite{nn14,nn15}
\begin{align*}
	J(u)=\sup_{p} \langle p, \nabla^{\alpha}u\rangle-J^{\ast}(p),
\end{align*}
where:
\begin{align*}
	J^{\ast}(p)= \begin{cases}
		0,      & \text{if } |p| \leq 1, \\[6pt]
		\infty, & \text{if } |p| > 1.
	\end{cases}
\end{align*}
\end{lem}

By using Lemma \ref{lem2}, the dual problem of
\eqref{eq5}, is obtained as:
\begin{align}\label{eq14}
\big(m^{(j+1)}_3, p^{(j+1)}\big)=
\arg \min_{m_3}\max_{p \in \Omega} \frac{\lambda_2}{\mu^{(j)}_3}
\langle  p, \nabla^{\alpha}m_3\rangle
+\frac{1}{2} \|m_3-\xi_6\|^2_2,
\end{align}
where $\Omega=\big\{p\in \mathbb{R}^{2nm}|p_i\in \mathbb{R}^2, \|p_i\|_2 \leq 1, \forall i \in \{1,\cdots,nm\}\big\}$. 
The solution of the above dual problem can be obtained by the following iterative scheme
\begin{align*}
& p^{(j+1)}=\arg \max_{p \in \Omega} \frac{\lambda_2}{\mu_3^{(j)}}
\langle p,\nabla^{\alpha} \hat{m}^{(j)}_3 \rangle
-\frac{1}{2\nu}\|p-p^{(j)}\|^2_2,\\
& m^{(j+1)}_3=\arg \min_{m_3} \frac{\lambda_2}{\mu_3}
\langle p^{(j+1)},\nabla^{\alpha} m_3\rangle
+\frac{1}{2}\|m_3-\xi_6\|^2_2+\frac{1}{2\tau}
\| m_3-\hat{m}_3^{(j)}\|^2_2,\\
&\hat{m}^{(j+1)}_3=2m^{(j+1)}_3-m^{(j)}_3,
\end{align*}	
where $\tau$ and $\nu$ are step size and a positive constant, respectively. After simplifying, we obtain:
\begin{align*}
&p^{(j+1)}=\frac{p^{(j)}+\frac{\lambda_2 \nu}{\mu_3}\nabla^{\alpha}\hat{m}^k_3}{
	\max\Big(\big|p^{(j)}+\frac{\lambda_2 \nu}{\mu_3} \nabla^{\alpha}\hat{m}^k_3\big|,1\Big)},\\
&m^{(j+1)}_3=m^{(j)}_3-\tau \Big( \frac{\lambda_2}{\mu_3}(\nabla^{\alpha})^Tp^{(j+1)}+m^{(j)}_3-\xi_6\Big).
\end{align*}

After calculating the solutions to the subproblems in the previous discussion, the 
$\boldsymbol{n}^{(j+1)}=\{n^{(j+1)}_i\}^3_{i=1}$ and $\boldsymbol{\mu}^{(j+1)}=\{\mu^{(j+1)}_i\}^3_{i=1}$
are updated in the final stage of the iterative method as follows:
\begin{align}\label{eq10}
&n^{(j+1)}_1=n^{(j)}_1+\mu^{(j)}_1(Ku^{(j+1)}-f-m^{(j+1)}_1),\\\label{eq12}
&n^{(j+1)}_2=n^{(j)}_2+\mu^{(j)}_2(Wu^{(j+1)}-m_2^{(j+1)}),\\\label{eq15}
&n^{(j+1)}_3=n^{(j)}_3+\mu^{(j)}_3(u^{(j+1)}-m_3^{(j+1)}),\\ \label{eq16}
& \mu^{(j+1)}_i=\sigma_i\mu^{(j)}_i,~i=1,2,3,
\end{align}
where $\sigma_i>1,~i=1,2,3,$  are considered as fixed constants. 
In continuation of the paper, for simplicity, we use 
$\boldsymbol{\sigma}=\{\sigma_i\}^3_{i=1}$.
The general structure of the proposed algorithm is presented in Algorithm \ref{alg1}.

\begin{algorithm}[!ht]
	\caption{Proposed Numerical Algorithm for Impulse Noise Grayscale Image Restoration
	}
	\label{alg1}
	\begin{algorithmic}[1]
		\State \textbf{Input:} Blurred and noisy image $f$ and blur kernel $K$
        \State \textbf{Parameters:} Set $\lambda_1,\lambda_2,\gamma, \alpha, \eta, \tau, \nu, \boldsymbol{\mu}^{0}, \boldsymbol{\sigma}$
        \State \textbf{Initialize:} $u^{(0)}=f, \boldsymbol{m}^{0}=\boldsymbol{n}^{0}=0, \rho^{(0)}=0, \hat{m}_3=f$
        \State $j \gets 0$
        \Repeat
		\State Compute $u^{(j+1)}$ by solving \eqref{eq8}
		\State Compute $m^{(j+1)}_1$ by using \eqref{eq9}
		\State Compute $n^{(j+1)}_1$ by using \eqref{eq10}
		\State Compute $m^{(j+1)}_2$ by using \eqref{eq11}
        \State Compute $n^{(j+1)}_2$ by using \eqref{eq12}
		\State Compute $m^{(j+1)}_3$ by solving the dual problem \eqref{eq14}
        \State Compute $n^{(j+1)}_3$ by using \eqref{eq15}
        \State Update $\boldsymbol{\mu}^{(j+1)}$ by using \eqref{eq16}
		\State $j \gets j+1$
       \Until{a stopping criterion is satisfied}
       \State \textbf{Output:} Deblurred and denoised image $u$
		
	\end{algorithmic}
\end{algorithm}

\subsection{Extension to Color Image Denoising}
In this subsection, we extend the method from the previous section to color 
images. In the first step, we expand the symbols used in the previous subsection for color images as follows:
\[u=\{u_{r},u_{g},u_{b}\}, \quad f=\{f_{r},f_{g},f_{b}\},\quad m_i=\{m_{i,r},m_{i,g},m_{i,b}\},
\quad n_i=\{n_{i,r},n_{i,g},n_{i,b}\},~~i=1,2,3.\]
In the color images deblurring, the blurred matrix
is considered as $K \in \mathbb{R}^{3nm \times 3nm}$ where
\cite{nn04}
\begin{align*}
	K= \left[ {\begin{array}{*{20}{c}}
			{\begin{array}{*{20}{c}}
					K_{rr}  \\
					K_{gr}  \\
					K_{br}  \\
			\end{array}} & {\begin{array}{*{20}{c}}
					K_{rg}  \\
					K_{gg}  \\
					K_{bg}  \\
			\end{array}} & {\begin{array}{*{20}{c}}
					K_{rb}  \\
					K_{gb}  \\
					K_{bb}  \\
			\end{array}}  \\
	\end{array}} \right].
\end{align*}
Using the periodic condition and the optimal condition for the 
 $u$ subproblem, we find:
\begin{align}\label{color1}
	\big(\mu_1 \Lambda^{\ast}\Lambda+(\mu^{(j)}_2+\mu^{(j)}_3)I\big)\hat{F}u^{(j+1)}
	=\mu^{(j)}_1 \Lambda^{\ast}\hat{F}\xi_1+\mu^{(j)}_2\hat{F}W^T\xi_2+\mu^{(j)}_3\hat{F}\xi_3,
\end{align}
where
$\hat{F}:=F\otimes I$, with $\otimes$ and $F$ denoting  the Kronecker product and 
the Fourier transform, respectively.
Also, $\Lambda$ is given by the Fourier transform decomposition as
\begin{align*}
	\Lambda = \left[ {\begin{array}{*{20}{c}}
			{\begin{array}{*{20}{c}}
					\Lambda_{rr}  \\
					\Lambda_{gr}  \\
					\Lambda_{br}  \\
			\end{array}} & {\begin{array}{*{20}{c}}
					\Lambda_{rg}  \\
					\Lambda_{gg}  \\
					\Lambda_{bg}  \\
			\end{array}} & {\begin{array}{*{20}{c}}
					\Lambda_{rb}  \\
					\Lambda_{gb}  \\
					\Lambda_{bb}  \\
			\end{array}}  \\
	\end{array}} \right],
\end{align*}
where $\Lambda_{ij}, i,j=r,g,b$ are obtained as diagonal matrices that depend on the eigenvalues of $K_{ij}, i,j=r,g,b$.
The other subproblems are solved in a manner similar to the topics discussed in the previous subsection.
The general structure of the proposed algorithm for color images is presented in Algorithm \ref{alg2}.
In this algorithm, it should be noted that the $\xi_{4,l}$ and $\xi_{5,l}$ values are the same as $\xi_{4}$ and $\xi_{5}$  from the previous sections, extended for each of the image color layers.
\begin{algorithm}[!ht]
	\caption{Proposed Numerical Algorithm for Impulse Noise Color Image Restoration
	}
	\label{alg2}
	\begin{algorithmic}[1]
		\State \textbf{Input:} Blurred and noisy image $f$ and blur kernel $K$
		\State \textbf{Parameters:} Set $\lambda_1,\lambda_2,\gamma, \alpha, \eta, \tau, \nu, \boldsymbol{\mu}^{0}, \boldsymbol{\sigma}$
		\State \textbf{Initialize:} $u^{(0)}=f, \boldsymbol{m}^{0}=\boldsymbol{n}^{0}=0, \rho^{(0)}=0, \hat{m}_3=f$
		\State $j \gets 0$
		\Repeat
		\State Compute $u^{(j+1)}$ by solving \eqref{color1}
		\For{l=r,g,b}
		\State Compute  $m^{(j+1)}_{1,l}=\psi_{ 1/\mu_1^{(j)}  }^{\gamma,\eta}(\xi_{4,l})$
		\State Compute  $n^{(j+1)}_{1,l}=n^{(j)}_{1,l}+\mu^{(j)}_{1,l}(Ku^{(j+1)}_{l}-f_{l}-m^{(j+1)}_{1,l})$
		\State Compute  $m^{(j+1)}_{2,l}=\psi_{ \lambda_1/\mu_2^{(j)}  }^{\gamma,\eta}(\xi_{5,l})$
		\State Compute $n^{(j+1)}_{2,l}=n^{(j)}_{2,l}+\mu^{(j)}_{2,l}(Wu^{(j+1)}-m_{2,l}^{(j+1)})$
		\State Compute $m^{(j+1)}_{3,l}$ by solving  the dual problem  for each of the color layers
		\State Compute $n^{(j+1)}_{3,l}=n^{(j)}_{3,l}+\mu^{(j)}_{3,l}(u^{(j+1)}-m_{3,l}^{(j+1)})$
		\State Update $\mu^{(j+1)}_{i,l}=\sigma_i\mu^{(j)}_{i,l},~i=1,2,3$
		\EndFor
		\State $j \gets j+1$
		\Until{a stopping criterion is satisfied}
		\State \textbf{Output:} Deblurred and denoised image $u$
		
	\end{algorithmic}
\end{algorithm}

\section{Convergence Analysis}\label{sec4}
In this section, we study the convergence of the proposed algorithm. If the proposed model 
had a convex structure, the convergence of the ADMM method could be readily established. 
Similarly, if the proposed model were coercive, we could determine convergence using an 
approach analogous to that of \cite{nn17}. However, the nonconvexity and noncoercivity of 
the proposed model necessitate a different approach to analyzing convergence. In the 
following section, we study the convergence of the proposed method using a technique 
similar to the one presented in \cite{nn7,nn18}.\\

\begin{lem}\label{lmc1}
Let $\{u^{(j)}, \boldsymbol{m}^{(j)}, \boldsymbol{n}^{(j)}\}$ be a sequence generated 
by the proposed algorithm. Then $\{n_i^{(j)}\}_j,~(i=1,2,3,4)$ and $\{L^{(j)}\}_j$ are bounded.
\end{lem}
\begin{proof}
Applying the first-order optimality condition to the $m_1$ subproblem yields
\begin{align*}
0 \in \frac{1}{\mu^{(j)}_1}\partial \| m^{(j+1)}_1 \|_{\mathrm{MCP}}+\big(m_1^{(j+1)}
-(Ku^{(j+1)}-f+\frac{1}{\mu^{(j)}_1}n^{(j)}_1)\big).
\end{align*}
Now by \eqref{eq10} and above formulae, we obtain:
\begin{align*}
n^{(j+1)}_1\in \partial \| m^{(j+1)}_1 \|_{\mathrm{MCP}}.
\end{align*}
Since the subgradient of MCP function is bounded, then 
$\|n^{(j+1)}_1\|_{\infty} \leq \gamma$. Similarly for $m_2$ subproblem, we get:
\begin{align*}
n^{(j+1)}_2 \in \lambda_1 \partial \|m^{(j+1)}_2\|_\mathrm{MCP}.
\end{align*}
Then by the boundedness of the MCP function, we can say 
$\|n^{(j+1)}_2\|_{\infty} \leq \lambda_1 \gamma$. 
In the next step, in order to check the boundedness of the $n^{(j+1)}_3$, by using othe first-order optimality condition 
of the \eqref{eq2}, we obtain:
\begin{align}\label{eq17}
K^Tn^{(j+1)}_1+W^Tn^{(j+1)}_2+n^{(j+1)}_3=0.
\end{align}
By taking infity norm of \eqref{eq17}, and applying the trianle inequality, we obtain
\begin{align*}
\|n^{(j+1)}_3\|_{\infty} \leq \|K^T\|_{\infty}\|n^{(j+1)}_1\|_{\infty}+\|W^T\|_{\infty}\|n^{(j+1)}_2\|_\infty.
\end{align*}
Therefore, the above relation implies the boundedness of $n^{(j+1)}_3$.\\
Now, we study the boundedness of $L^{(j)}$. 
By considering the structure of $L^{(j)}$, we can write:
\begin{align}\nonumber
L(u^{(j+1)},\boldsymbol{m}^{(j+1)},\boldsymbol{n}^{(j+1)},\boldsymbol{\mu}^{(j+1)})
	&-L(u^{(j+1)},\boldsymbol{m}^{(j+1)},\boldsymbol{n}^{(j+1)},\boldsymbol{\mu}^{(j)})\\\label{eq18}
&=\sum_{i=1}^{3}\frac{\mu^{(j+1)}_i-\mu^{(j)}_i}{2(\mu^{(j)})^2}\|n^{(j+1)}_i-n^{(j)}_i\|^2_2,
\end{align}
\begin{align}\nonumber
	L(u^{(j+1)},\boldsymbol{m}^{(j+1)},\boldsymbol{n}^{(j+1)},\boldsymbol{\mu}^{(j)})
	&-L(u^{(j+1)},\boldsymbol{m}^{(j+1)},\boldsymbol{n}^{(j)},\boldsymbol{\mu}^{(j)})\\\label{19}
	&=\sum_{i=1}^{3}\frac{1}{\mu^{(j)}}\|n^{(j+1)}_i-n^{(j)}_i\|^2_2.
\end{align}
Considering the subproblems used for $m^{(j+1)}_i,~i=1,2,3$, we have:
\begin{align}\label{20}
&L(u^{(j+1)},m^{(j+1)}_1,m^{(j+1)}_2,m^{(j+1)}_3,\boldsymbol{n}^{(j)},\boldsymbol{\mu}^{(j)})
-L(u^{(j+1)},m^{(j+1)}_1,m^{(j+1)}_2,m^{(j)}_3,\boldsymbol{n}^{(j+1)},\boldsymbol{\mu}^{(j)}) \leq 0,\\\label{21}
&L(u^{(j+1)},m^{(j+1)}_1,m^{(j+1)}_2,m^{(j)}_3,\boldsymbol{n}^{(j)},\boldsymbol{\mu}^{(j)})
-L(u^{(j+1)},m^{(j+1)}_1,m^{(j)}_2,m^{(j)}_3,\boldsymbol{n}^{(j+1)},\boldsymbol{\mu}^{(j)}) \leq 0,\\\label{22}
&L(u^{(j+1)},m^{(j+1)}_1,m^{(j)}_2,m^{(j)}_3,\boldsymbol{n}^{(j)},\boldsymbol{\mu}^{(j)})
-L(u^{(j+1)},m^{(j)}_1,m^{(j)}_2,m^{(j)}_3,\boldsymbol{n}^{(j+1)},\boldsymbol{\mu}^{(j)}) \leq 0.
\end{align}
By  taking the second derivative of the objective function of $u$ subproblem, we obtain the Hessian matrix as :
\begin{align*}
\mathbf{H}=\mu_1K^TK+\mu_2W^TW+\mu_3I.
\end{align*}
Therefore, we get:
\begin{align*}
x^T \mathbf{H} x= \mu_1 \|Kx\|^2_2+\mu_2\|Wx\|^2_2+\mu_3 \|x\|^2_2>0,\quad \forall x \in \mathbb{R}^n \setminus \{0\}.
\end{align*}
So, this Hessian matrix is strictly positive definite. 
Then the $u$ subproblem is strongly convex. Consequently, we can find 
a positive constant $C>0$ that

\begin{align}\label{eq24}
L(u^{(j+1)},\boldsymbol{m}^{(j)},\boldsymbol{n}^{(j)},\boldsymbol{\mu}^{(j)})
-L(u^{(j)},\boldsymbol{m}^{(j)},\boldsymbol{n}^{(j)},\boldsymbol{\mu}^{(j)})
\leq -\frac{C}{2}\|u^{(j)}-u^{(j)}\|^2_2.
\end{align}
By combining \eqref{eq18}-\eqref{eq24}, we obtain:
\begin{align}\nonumber
	L(u^{(j+1)},\boldsymbol{m}^{(j+1)},\boldsymbol{n}^{(j+1)},\boldsymbol{\mu}^{(j+1)})
	&-L(u^{(j)},\boldsymbol{m}^{(j)},\boldsymbol{n}^{(j)},\boldsymbol{\mu}^{(j)})\\ \label{eq25}
	&\leq \sum_{i=1}^{3}\frac{\mu_i^{(j+1)}+\mu_i^{(j)}}{2(\mu_i^{(j)})^2}
	\|n^{(j+1)}_i-n^{(j)}\|^2_2
	 -\frac{C}{2}\|u^{(j)}-u^{(j)}\|^2_2.
\end{align}
In the first steps of the proof, we showed that 
$\{n^{j}_i\}_j,~(i=1,2,3)$ are bounded. Then for $i=1,2,3$, we can find $M_i$ such that
$\|n^{j+1}_i-n^{(j)}_i\|^2_2\leq M_i$. So, by considering \eqref{eq16} and \eqref{eq25}, we get:
\begin{align}\nonumber
L(u^{(j+1)},\boldsymbol{m}^{(j+1)},\boldsymbol{n}^{(j+1)},\boldsymbol{\mu}^{(j+1)})
&-L(u^{(0)},\boldsymbol{m}^{(0)},\boldsymbol{n}^{(0)},\boldsymbol{\mu}^{(0)})\\ \label{eq27}
&\leq \sum_{i=1}^{3}\Big[\frac{(\sigma_i+1)M_i}{2\mu_i^{(0)}}
\big(\frac{1-(\frac{1}{\sigma_i})^j}{1-\frac{1}{\sigma_i}}\big) \Big ]
-\frac{C}{2}\sum^{j}_{i=0}\|u^{(i+1)}-u^{(i)}\|^2_2.
\end{align}
Since $\sigma_i>1, (i=1,2,3)$, when $j\rightarrow \infty$, we get:
\begin{align}\label{eq19}
\lim_{j\rightarrow \infty}&L(u^{(j+1)},\boldsymbol{m}^{(j+1)},\boldsymbol{n}^{(j+1)},\boldsymbol{\mu}^{(j+1)})
	\leq L(u^{(0)},\boldsymbol{m}^{(0)},\boldsymbol{n}^{(0)},\boldsymbol{\mu}^{(0)})+
	 \sum_{i=1}^{3}\Big[\frac{\sigma_i(\sigma_i+1)M_i}{2(\sigma_i-1)\mu_i^{(0)}}
	 \Big ]< \infty.
\end{align}
Therefore, the sequence $\{L^{(j)}\}_j$ is bounded above. Next, we show that it is also bounded below.
By using properties of norm, we have:
\begin{align}\nonumber
L(u^{(j+1)},\boldsymbol{m}^{(j+1)}&,\boldsymbol{n}^{(j+1)},\boldsymbol{\mu}^{(j+1)})
=\|m_1^{(j+1)}\|_\mathrm{MCP}+\lambda_1 \| m_2\|_\mathrm{MCP}+\lambda_2 \|m_3\|_\mathrm{FTV}\\\nonumber
&+\frac{\mu^{(j+1)}_1}{2}\|Ku^{(j+1)}-f-m^{(j+1)}_1+\frac{1}{\mu^{(j+1)}_1}n^{(j+1)}_1\|^2_2-\frac{\|n^{(j+1)}_1\|^2_2}{2\mu^{(j+1)}_1}\\ \nonumber
&+\frac{\mu^{(j+1)}_2}{2}\|Wu^{(j+1)}-m^{(j+1)}_2+\frac{1}{\mu^{(j+1)}_2}n^{(j+1)}_2\|^2_2-\frac{\|n^{(j+1)}_2\|^2_2}{2\mu^{(j+1)}_2}\\
\label{eq20}
&+\frac{\mu^{(j+1)}_3}{2}\|u^{(j+1)}-m^{(j+1)}_3+\frac{1}{\mu^{(j+1)}_3}n^{(j+1)}_3\|^2_2-\frac{\|n^{(j+1)}_3\|^2_2}{2\mu^{(j+1)}_3}.
\end{align}
Then, by taking $j\rightarrow \infty$ and considering that the values of the norms are positive, we obtain
\begin{align*}
\lim_{j\rightarrow \infty}
L(u^{(j+1)},\boldsymbol{m}^{(j+1)}&,\boldsymbol{n}^{(j+1)},\boldsymbol{\mu}^{(j+1)})
\geq \lim_{j\rightarrow \infty} -\sum^{3}_{i=1}\frac{\|n^{(j+1)}_i\|^2_2}{2\mu^{(j+1)}_i}>-\infty.
\end{align*}
That completes the proof.
\end{proof}
\begin{lem}\label{lmc2}
The sequences $\{u^{(j)}\}$ and $\{m^{(j)}_i\},~(i=1,2,3)$ generated by algorithm are bounded if the 
MCP parameters $\gamma$, $\eta$ and regularization parameter 
$\lambda$ satisfies:
\begin{align}\label{eq26}
L(u^{(0)},\boldsymbol{m}^{(0)},\boldsymbol{n}^{(0)},\boldsymbol{\mu}^{(0)})+
\sum_{i=1}^{3}\Big[\frac{\sigma_i(\sigma_i+1)M_i}{2(\sigma_i-1)\mu_i^{(0)}}
\Big ]<\frac{1}{2}\gamma^2 \eta.
\end{align}
\end{lem} 
\begin{proof}
Since $\lim_{j\rightarrow \infty}\frac{\|n^{(j+1)}_i\|^2_2}{2\mu^{(j+1)}_i}=0$, then by using
\eqref{eq19} and \eqref{20} and lemma's assumption, we get
\begin{align}\label{eq21}
\lim_{j \rightarrow \infty }\|m^{(j+1)}_1\|_\mathrm{MCP} \leq  \lim_{j \rightarrow \infty }
L(u^{(j+1)},\boldsymbol{m}^{(j+1)},\boldsymbol{n}^{(j+1)},\boldsymbol{\mu}^{(j+1)}) < \frac{1}{2}\gamma^2 \eta.
\end{align} 
Now let $\{m_1^{(j)}\}$ be unbounded. Therefore, when 
$j \rightarrow \infty$, we can find at least one coordinate component
$i$ such that $|(m_1^{(j+1)})_i |\rightarrow \infty$. The, using the definition 
of MCP function, we can say that 
$h_{\gamma,\lambda}((m_1^{(j+1)})_i)=\frac{1}{2}\gamma^2 \eta$. This is due to:
\begin{align*}
\| m_1^{(j+1)})\|_\mathrm{MCP} \geq h_{\gamma,\lambda}((m_1^{(j+1)})_i)=\frac{1}{2}\gamma^2 \eta.
\end{align*}
By taking the limit we have
\begin{align*}
\lim_{j \rightarrow \infty } \| m_1^{(j+1)})\|_\mathrm{MCP} \geq \frac{1}{2}\gamma^2 \eta,	
\end{align*}
where
contradicts with the upper boundary condition \eqref{eq21}.
The boundedness of $\{u^{(j)}\}$ and $\{m^{(j)}_i\},~(i=2,3)$ follows from 
\eqref{eq10}-\eqref{eq15} and the boundedness of $\{n^{(j)}_i\}$ and $\{m^{(j)}_1\}$.
\end{proof}

\begin{thm}
Let the MCP parameters 
 $\gamma$ and $\eta$ satisfy  \eqref{eq26}. Then
sequence $\{X^{(j)}\}:=\{u^{(j)}, \boldsymbol{m}^{(j)},\boldsymbol{n}^{(j)}\}$
is bounded and the following hold:
\begin{align*}
\lim_{j\rightarrow \infty} \|u^{(j+1)}-u^{(j)}\|_2=0,
\quad
\lim_{j\rightarrow \infty} \|m_i^{(j+1)}-m_i^{(j)}\|_2=0,
\quad 
i=1,2,3.
\end{align*}
In addition, there exists a subsequence of $\{X^{j}\}$ that converges to a stationary 
point 
$\{X^{\ast}\}:=\{u^{\ast}, \boldsymbol{m}^{\ast},\boldsymbol{n}^{\ast}\}$.
\end{thm}
\begin{proof}
In the first part of the proof, the boundedness of the sequence $\{X^{(j)}\}$, follows from Lemmas \ref{lmc1} and \ref{lmc2}.
In the next step, using \eqref{eq27}, we get:
\begin{align*}
0\leq\frac{C}{2}\lim_{j \rightarrow \infty }\sum^{j}_{i=0}\|u^{(i+1)}-u^{(i)}\|^2_2 &
\leq L(u^{(0)},\boldsymbol{m}^{(0)},\boldsymbol{n}^{(0)},\boldsymbol{\mu}^{(0)})\\
&-\lim_{j \rightarrow \infty }L(u^{(j+1)},\boldsymbol{m}^{(j+1)},\boldsymbol{n}^{(j+1)},\boldsymbol{\mu}^{(j+1)})
\\
&+ \lim_{j \rightarrow \infty }\sum_{i=1}^{3}\Big[\frac{(\sigma_i+1)M_i}{2\mu_i^{(0)}}
\big(\frac{1-(\frac{1}{\sigma_i})^j}{1-\frac{1}{\sigma_i}}\big) \Big ]
< \infty.
\end{align*}
This is due to
\begin{align}\label{eq30}
\lim_{j \rightarrow \infty }\|u^{(i+1)}-u^{(i)}\|_2=0.
\end{align}
Considering the boundness of $\{n^{(j)_i}\},~(i=1,2,3)$ and using  \eqref{eq10} to \eqref{eq16}, we have
\begin{align}\label{eq29}
&\lim_{j \rightarrow \infty } \|Ku^{(j+1)}-f-m^{(j+1)}_1\|_2=\lim_{j \rightarrow \infty } \frac{\|n^{(j+1)}_1-n^{(j)}_1\|_2}{\mu^{j+1}_1}=0,\\\label{eq31}
&\lim_{j \rightarrow \infty } \|Wu^{(j+1)}-m^{(j+1)}_2\|_2=\lim_{j \rightarrow \infty } \frac{\|n^{(j+1)}_2-n^{(j)}_2\|_2}{\mu^{j+1}_2}=0,\\\label{eq32}
&\lim_{j \rightarrow \infty } \|u^{(j+1)}-m^{(j+1)}_3\|_2=\lim_{j \rightarrow \infty } \frac{\|n^{(j+1)}_3-n^{(j)}_3\|_2}{\mu^{j+1}_3}=0.
\end{align}
From the triangle inequality, we
obtain
\begin{align}\nonumber
0\leq\|m_1^{(j+1)}-m_1^{(j)}\|_2 \leq & \|Ku^{(j+1)}-f-m^{(j+1)}_1\|_2+\|Ku^{(j)}-f-m^{(j)}\|_2\\ \label{eq28}
&+\|K(u^{(j+1)}-u^{(j)})\|_2.
\end{align}	
Then, based on relations \eqref{eq30}, \eqref{eq29} and \eqref{eq28}, we can say that 
$\|m_1^{(j+1)}-m_1^{(j)}\|_2\rightarrow0$ as $j \rightarrow \infty$.
In a similar way, we can prove that for $i=2,3$,
$\|m^{(j+1)}_i-m^{(j)}_i\|_2\rightarrow 0$ as $j \rightarrow \infty$.\\
Since the sequence $\{X^j\}$ is bounded, there exists a convergent subsequence $\{X^{r_j}\}$ that converges to $\{X^\ast\}$. Also, since 
$L$ is the lower semicontinuous, we get
\begin{align*}
\lim_{j\rightarrow \infty} \inf
L(u^{(r_j+1)},m^{(r_j+1)}_1,m^{(k_j+1)}_2,
m^{(k_j)}_3,\boldsymbol{n}^{(k_j)},\boldsymbol{\mu}^{(k_j)}
)\geq 
L(X^{\ast},\boldsymbol{\mu}^{\infty}).
\end{align*}
Also, $m^{(k_j+1)}_2$ is obtained by minimizing $L$ with respect to
$m_1$, then we have 
\begin{align*}
	\lim_{j\rightarrow \infty} \sup
	L(u^{(r_j+1)},m^{(r_j+1)}_1,m^{(k_j+1)}_2,
	m^{(k_j)}_3,\boldsymbol{n}^{(k_j)},\boldsymbol{\mu}^{(k_j)}
	)\leq
	L(X^{\ast},\boldsymbol{\mu}^{\infty}).
\end{align*}
Then based on above relations, we can write:
\begin{align*}
	\lim_{j\rightarrow \infty}
	L(u^{(r_j+1)},m^{(r_j+1)}_1,m^{(k_j+1)}_2,
	m^{(k_j)}_3,\boldsymbol{n}^{(k_j)},\boldsymbol{\mu}^{(k_j)}
	)=
	L(X^{\ast},\boldsymbol{\mu}^{\infty}).
\end{align*}
Therefore we obtain:
\begin{align*}
&\lim_{j\rightarrow \infty}
\|m^{(r_j+1)}_1\|_{\mathrm{MCP}}=\|m^{\ast}_1 \|_{\mathrm{MCP}},
\quad
\lim_{j\rightarrow \infty}
\|m^{(r_j+1)}_2\|_{\mathrm{MCP}}=\|m^{\ast}_2 \|_{\mathrm{MCP}},
\\
&\lim_{j\rightarrow \infty}
\|m^{(r_j)}_3\|_{\mathrm{MCP}}=\|m^{\ast}_3 \|_{\mathrm{MCP}}.
\end{align*}
By using above limits and relations \eqref{eq29}-\eqref{eq32}, we obtain
\begin{align*}
Ku^{\ast}-f-m^{\ast}_1=0,\quad Wu^{\ast}-m^{\ast}_2=0, \quad u^{\ast}-m^{\ast}_3=0.
\end{align*}
Also by taking limit of the first-order optimality conditions of the auxiliary subproblems as $j\rightarrow \infty$, we get
\begin{align*}
& 0\in \partial \| m^{\ast}_1 \|_{\mathrm{MCP}}-n^{\ast}_1,
\quad
0 \in \lambda_1 \partial \|m^{\ast}_2\|_\mathrm{MCP}-n^{\ast}_2,
\quad
K^Tn^{\ast}_1+W^Tn^{\ast}_2+n^{\ast}_3=0.
\end{align*}
Also, we obtain
\begin{align*}
0 \in K^T\partial \|Ku^{\ast}-f\|_{\mathrm{MCP}}+\lambda_1W^T \partial\|Wu^{\ast}\|_{\mathrm{MCP}}+\lambda_2 \partial \|u^{\ast}\|_{\mathrm{FTV}}.
\end{align*}
Based on the above discussion, the Karush--Kuhn--Tucker (KKT) conditions of the optimization problem hold at the limit point $X^{\ast}$.
\end{proof}	

\section{Simulation and Analysis Results} \label{sec5}
In this section, we present and simulate the proposed algorithm for several examples.
For this purpose, a system with the following specifications is used:
MATLAB 2014, Windows 10 (64-bit), with an Intel Core i3-5005U CPU @ 2.00 GHz.
The images used in the examples are taken from the USC-SIPI image database\footnote{\url{https://sipi.usc.edu/database/}}.
To add noise to the image, we used the ``salt \& pepper'' and ``Random-valued'' noise types.
Since the noise generated in the simulation is random, different images may result on each run;
 to ensure reproducibility, we set the seed using the MATLAB command \texttt{rng(42)}.
In this paper, we use $AB(n)$ to denote the average PSF of size $n$, $MB(n,\theta)$ to denote the motion PSF of size  $n$ and angle 
 $\theta$, and
 $GB(n, \sigma)$ to denote the Gaussian PSF of size $n$ and $\sigma$ with standard deviation.
 It should be noted that to generate these PSF, we use the function \texttt{fspecial} provided by MATLAB.
 The proposed algorithm depends on various parameters. In the numerical results section, we select 
 $\alpha \in \{0.75,1,1.25,1.75\}$, $\mu_1 \in \{5,15,25,50,100\}$,
 $\mu_2,\mu_3, \gamma_1, \gamma_2\in \{1,0.1,0.01,0.001,0.0001\}$
 $\gamma\in {0.5,75}$,$\eta=3$
 $\sigma_i=1.01,i=1,2,3$,$\tau=0.02$ and $\nu=0.005$, and
 after calculations, we report the best outcome obtained as the final result. 
 
 We also consider the stopping condition for the presented algorithms as follows:
\begin{align*}
\frac{\|u^{(j+1)}-u^{(j)}\|_2}{\|u^{(j+1)} \|_2}\leq \mathrm{tol}.
\end{align*}
In all the examples, we set 
 $\mathrm{tol}=10^{-3}$.
 
Various metrics exist for comparing a reconstructed image with a noise-free image; in our results, we report the PSNR (Peak Signal-to-Noise Ratio), SNR (Signal-to-Noise Ratio), SSIM (Structural Similarity Index), and FSIM (Feature Similarity Index Measure).

\begin{example}\label{ex1}
	\normalfont
As the first test image, we consider \texttt{clock-5.1.12}, a 
$256 \times 256$-pixel grayscale image with 8-bit depth.
We use the $k_1=AB(5)$ and $k_2=MB(20,45)$ kernels for blurring.
In Table \ref{tab1}, the results obtained from the proposed method, by considering the salt \& pepper noise of various density, are compared with the algorithms presented in the papers by \cite{nn19} and \cite{nn20}.
Based on the results presented in this table, it can be observed that, although the proposed algorithm requires more time than the two comparison methods-a consequence of its inherent complexity-it yields better results.
Furthermore, Figure \ref{fig1} presents a comparison of the resulting output images. In this figure, the impact of the function and the noise on the image is clearly evident-to the extent that the original image has been obliterated and is no longer recognizable- whereas the proposed algorithm has successfully recovered the original image.
In this figure, the PSF is taken as $k_2$ and the salt \& pepper noise with a  density $10\% $ is used.
\end{example}

	\begin{figure}[ht!]
		\centering
		\makebox[0pt]{
			\subfigure[]{\label{fig:gull}\includegraphics[width=0.25\textwidth]{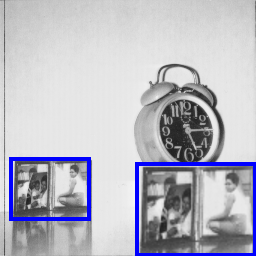}}
			\subfigure[]{\label{fig:tiger}\includegraphics[width=0.25\textwidth]{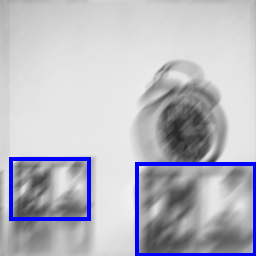}}
			\subfigure[]{\label{fig:tiger}\includegraphics[width=0.25\textwidth]{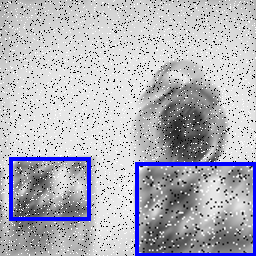}}
		}\\
		\vspace{-.2cm}
		\makebox[0pt]{
			\subfigure[]{\label{fig:gull}\includegraphics[width=0.25\textwidth]{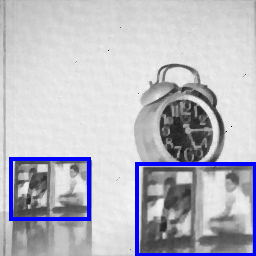}}
			\subfigure[]{\label{fig:tiger}\includegraphics[width=0.25\textwidth]{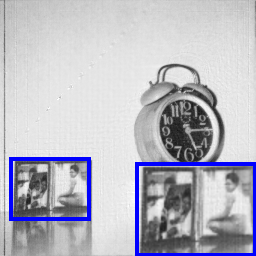}}
			\subfigure[] {\label{fig:tiger}\includegraphics[width=0.25\textwidth]{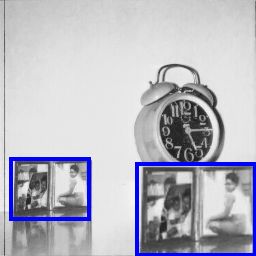}}
		}\\
		\vspace{-.3cm}
		\caption{(a) Clear image, (b) blurred image, (c) blurred and noisy image, and restored images by (d) the method in \cite{nn19}, (e) the method in \cite{nn20}, and (f) the proposed method.
			}
		\label{fig1}
	\end{figure}
	\begin{example}\label{ex2}
		\normalfont
		In the next example, we select
		\texttt{Female-4.1.01
		}, a 
		$256 \times 256$-pixel color image with 24 bits/pixel. First, we convert this image to grayscale using the \texttt{rgb2gray} command, and then we apply the proposed algorithm to it. For this example the PSF $MA(20,30)$ and the salt \& pepper noise with a  density $10\% $ are considered. 
		To compare the performance of the proposed algorithm with that of the methods in \cite{nn19} and \cite{nn20} over the same number of iterations, plots of the relative error and PSNR are presented in Figure \ref{fig2}. The results indicate that, for a given number of iterations, the proposed algorithm yields better results.
		To numerically examine the convergence of the proposed method, Figure \ref{fig3} has been plotted for this example. Based on the obtained results, the convergence of the proposed method can be observed.
		Additionally, Figure \ref{fig4} is presented to show the results obtained from restoring images corrupted by various types of noise.
	\end{example}

	\begin{center}
		\begin{table}[H]
			\caption{{ Numerical results for Example~\ref{ex1} with PSF $k_1$.}}
			\label{tab1}
			\centering
			\begin{tabular}{ccccccccccccccc}
				\hline
method   & density  &Iter &Time(s)    &PSNR    &SNR    &SSIM   &FSIM    \\
\hline
         &$10\%$    &61 &26.555   &40.608  &38.26 &0.98432 &0.98931\\
Proposed &$20\%$    &100 &43.551   &38.11  &35.761 &0.97892 &0.98381\\
         &$30\%$    &166 &72.214   &35.778 &33.430 &0.95107 &0.97283    \\
		 &$60\%$    &211&101.85   &26.156 &23.807 &0.83405 &0.84277    \\
				            \hline
         &$10\%$    &186 &4.0958   &36.919  &34.571 &0.94900 &0.96153\\				            
Method in \cite{nn19} &$20\%$    &60&1.3446&33.552&31.204&0.95802&0.95948\\
         &$30\%$    &69&1.4796&31.252&28.903&0.93941&0.93864    \\
		 &$60\%$    &202&4.4032&25.039 &22.691 &0.83319 &0.82609     \\
                             \hline
         &$10\%$    &127 &3.3321   &33.869  &31.521 &0.93782 &0.95901\\                             
Method in \cite{nn20} &$20\%$    &831 &19.435    &29.849  &27.500&0.84878 &0.90509\\
         &$30\%$    &947 &20.841    &29.070  &26.721&0.82911 &0.89210   \\
         &$60\%$    &1000&25.479    &24.675  &22.327&0.65569 &0.75837    \\
                             \hline
			\end{tabular}
		\end{table}
	\end{center}
\vspace{-1cm}	
\begin{figure}[H]
	\centering
	\makebox[0pt]{
		\subfigure[]{\label{fig:gull}\includegraphics[width=0.5\textwidth]{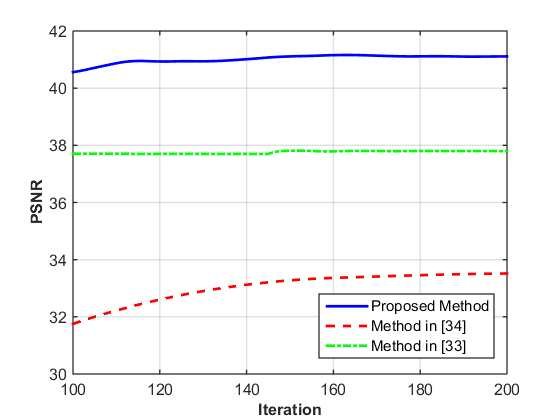}}
		\subfigure[]{\label{fig:tiger}\includegraphics[width=0.5\textwidth]{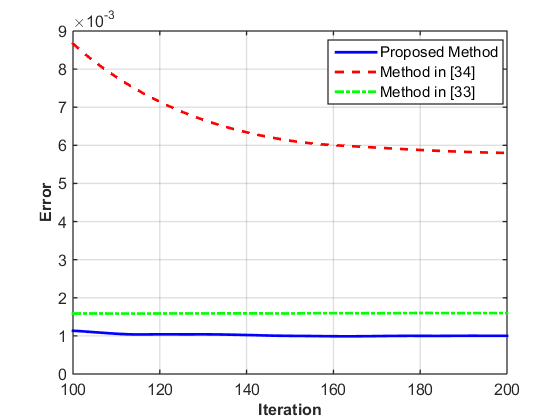}}
	}\\
	\vspace{-.3cm}
	\caption{Comparison results for (a) PSNR, (b) Relative Error.}
	\label{fig2}
\end{figure}	
\vspace{-1cm}	
\begin{figure}[H]
	\centering
	\makebox[0pt]{
		\subfigure[]{\label{fig:gull}\includegraphics[width=0.5\textwidth]{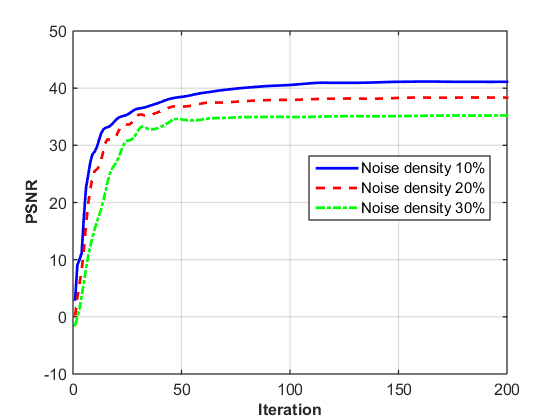}}
		\subfigure[]{\label{fig:tiger}\includegraphics[width=0.5\textwidth]{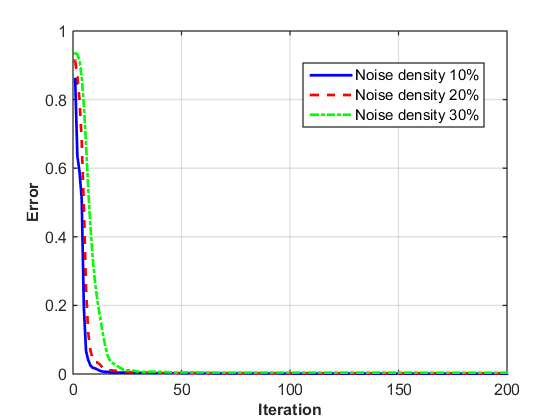}}
	}\\
	\vspace{-.3cm}
	\caption{Results for  different noisy density of the proposed method, (a) PSNR, (b) Relative error.}
	\label{fig3}
\end{figure}

	\begin{figure}[H]
	\centering
	\makebox[0pt]{
		\subfigure[]{\label{fig:gull}\includegraphics[width=0.25\textwidth]{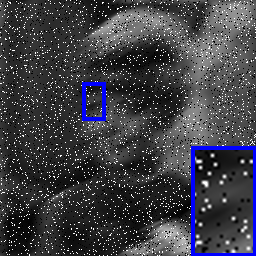}}
		\subfigure[]{\label{fig:tiger}\includegraphics[width=0.25\textwidth]{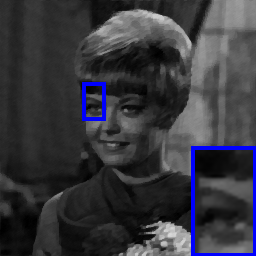}}
		\subfigure[]{\label{fig:tiger}\includegraphics[width=0.25\textwidth]{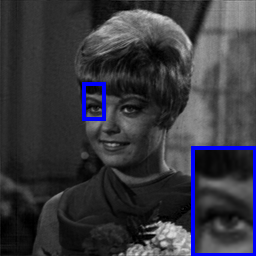}}
				\subfigure[]{\label{fig:tiger}\includegraphics[width=0.25\textwidth]{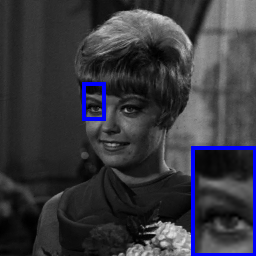}}
	}\\
	\vspace{-.2cm}
	\makebox[0pt]{
	\subfigure[]{\label{fig:gull}\includegraphics[width=0.25\textwidth]{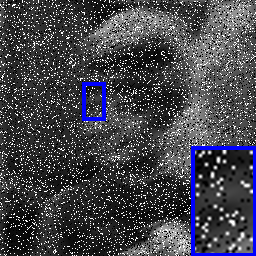}}
	\subfigure[]{\label{fig:tiger}\includegraphics[width=0.25\textwidth]{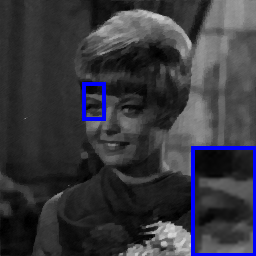}}
	\subfigure[]{\label{fig:tiger}\includegraphics[width=0.25\textwidth]{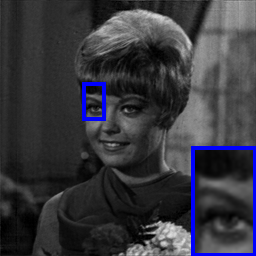}}
	\subfigure[]{\label{fig:tiger}\includegraphics[width=0.25\textwidth]{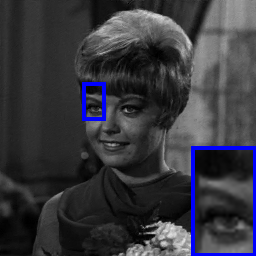}}
}\\
	\makebox[0pt]{
	\subfigure[]{\label{fig:gull}\includegraphics[width=0.25\textwidth]{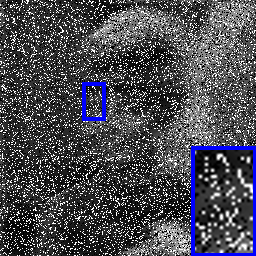}}
	\subfigure[]{\label{fig:tiger}\includegraphics[width=0.25\textwidth]{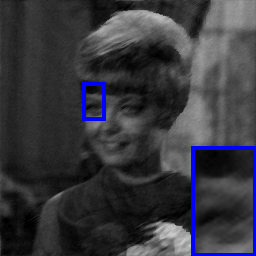}}
	\subfigure[]{\label{fig:tiger}\includegraphics[width=0.25\textwidth]{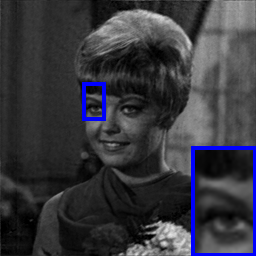}}
	\subfigure[]{\label{fig:tiger}\includegraphics[width=0.25\textwidth]{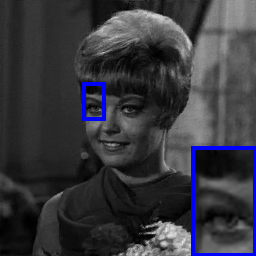}}
}
	\vspace{-.3cm}
	\caption{(a, e, i) Blurred and noisy images with noise densities of 
		$10\%, 
		20\%$ and 
		$30\%$, respectively; restored images obtained by (b, f, j) the method in \cite{nn19}, (c, g, k) the method in \cite{nn20}, and (d, h, l) the proposed method, each at the corresponding noise density.}
	\label{fig4}
\end{figure}	

\begin{example}\label{ex3}
	\normalfont	
In this example, we study the proposed algorithm for color images. For this purpose, we consider \texttt{House-4.1.05}, a 
$256 \times 256$-pixel color image with 24 bits/pixel, as the target image. Given that the algorithms proposed in references \cite{nn19} and \cite{nn20} are not designed for composite PSF, in this example, we apply a different blurring PSF to each layer of the image
as
\[K=\mathrm{diag}\big(AB(7),GA(11,5),MA(20,45)\big)\]
 and subsequently add random-valued noise with a different density. 
The results obtained are presented in Table \ref{tab2} and Figure \ref{fig5}. Based on these results, it can be observed that the numerical outcomes for the restoration of degraded images are satisfactory.
\end{example}	

\begin{center}
	\begin{table}[H]
		\caption{{ Numerical results for Example~\ref{ex3}.}}
		\label{tab2}
		\centering
		\begin{tabular}{ccccccccccccccc}
			\hline
			method   & density  &Iter &Time    &PSNR    &SNR    &SSIM   &FSIM    \\
			\hline
			&$10\%$     &130&244.1    &37.842 &33.228&0.99117 &0.97935 \\
			Proposed &$20\%$    &138&285.43    &36.888 &32.274&0.98874 & 0.97539\\
			&$30\%$    &181&289.6    &34.966 &30.352&0.98278 &0.96456    \\
			&$60\%$    &219&331.72    &25.85 &21.236&0.90005 &0.80558   \\
			\hline
			&$10\%$     &231&15.24    &34.16 &29.546&0.97949 &0.96086  \\				            
Method in \cite{nn19} &$20\%$   &84&    6.1039 &31.347&26.733  & 0.96657&0.92617 \\
			&$30\%$     &88& 6.4528    &30.121  &25.507& 0.95817&0.90848    \\
			&$60\%$    &295&23.133    &23.000 &18.386&0.8128 &0.75819   \\
			\hline
			&$10\%$     &54&7.2995    &31.251 & 26.637 &0.96867 &0.91688 \\                             
Method in \cite{nn20} &$20\%$   &136&13.73    &30.594 &25.98&0.96367 &0.90938 \\
			&$30\%$     &685&52.413    &29.466 &24.852&0.95091 & 0.90249   \\
			&$60\%$    &308&33.526    &21.047 &16.433&0.74924  &0.72576    \\
			\hline
		\end{tabular}
	\end{table}
\end{center}
	
\begin{example}
	\normalfont	
	As a final example, we select
	\texttt{Female-4.1.04}, a 
	$256 \times 256$-pixel color image with 24 bits/pixel.
	The PSF is chosen as
	\begin{align*}
		K= \left[ {\begin{array}{*{20}{c}}
				{\begin{array}{*{20}{c}}
						0.7 AB(15)  \\
						0.15GA(11,9)  \\
						0.15GA(20,17)   \\
				\end{array}} & {\begin{array}{*{20}{c}}
						0.1GA(20,11)   \\
						0.8AB(15)  \\
						0.1AB(11)  \\
				\end{array}} & {\begin{array}{*{20}{c}}
						0.2MA(40,90) \\
						0.2MA(20,30)  \\
						0.6MA(30,45)   \\
				\end{array}}  \\
		\end{array}} \right].
	\end{align*}
	
	Also, random-valued noise with a different density is added to images. 	
	Images restored using the proposed algorithm are presented in 
	Figure \ref{fig6} for noise of various density.	
	Furthermore, to examine the convergence of the proposed method, the PSNR and 
	relative error plots are presented in Figure \ref{fig7} for various iterations and noise 
	density. The results shown in the plots indicate that the method converges.	
\end{example}

\begin{figure}[H]
	\centering
	\makebox[0pt]{
		\subfigure[]{\label{fig:gull}\includegraphics[width=0.25\textwidth]{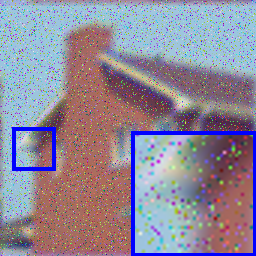}}
		\subfigure[]{\label{fig:tiger}\includegraphics[width=0.25\textwidth]{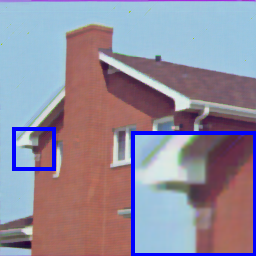}}
		\subfigure[]{\label{fig:tiger}\includegraphics[width=0.25\textwidth]{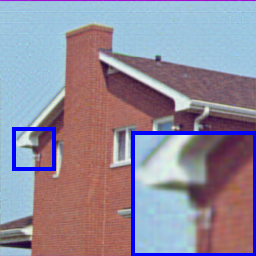}}
		\subfigure[]{\label{fig:tiger}\includegraphics[width=0.25\textwidth]{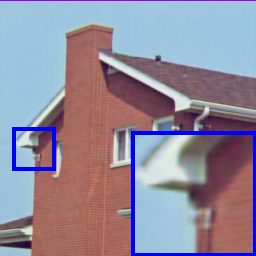}}
	}\\
	\vspace{-.2cm}
	\makebox[0pt]{
		\subfigure[]{\label{fig:gull}\includegraphics[width=0.25\textwidth]{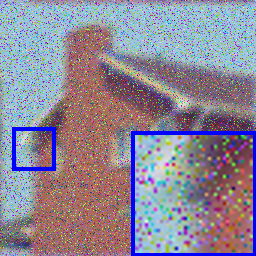}}
\subfigure[]{\label{fig:tiger}\includegraphics[width=0.25\textwidth]{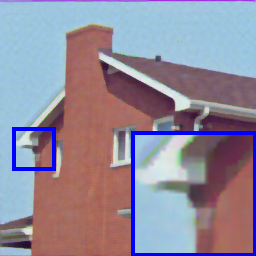}}
\subfigure[]{\label{fig:tiger}\includegraphics[width=0.25\textwidth]{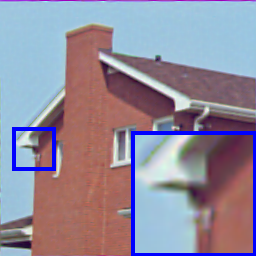}}
\subfigure[]{\label{fig:tiger}\includegraphics[width=0.25\textwidth]{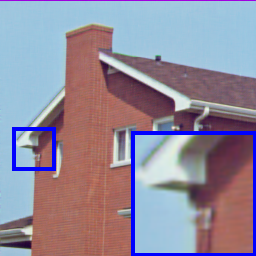}}
	}\\
	\makebox[0pt]{
		\subfigure[]{\label{fig:gull}\includegraphics[width=0.25\textwidth]{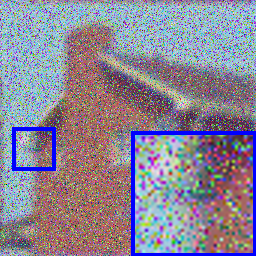}}
\subfigure[]{\label{fig:tiger}\includegraphics[width=0.25\textwidth]{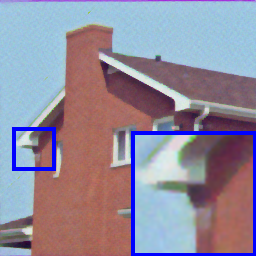}}
\subfigure[]{\label{fig:tiger}\includegraphics[width=0.25\textwidth]{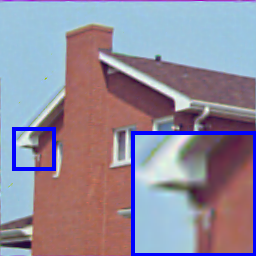}}
\subfigure[]{\label{fig:tiger}\includegraphics[width=0.25\textwidth]{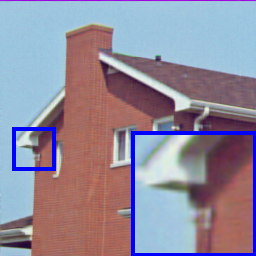}}
	}
	\vspace{-.3cm}
	\caption{(a, e, i) Blurred and noisy images with noise densities of 
		$10\%, 
		20\%$ and 
		$30\%$, respectively; restored images obtained by (b, f, j) the method in \cite{nn19}, (c, g, k) the method in \cite{nn20}, and (d, h, l) the proposed method, each at the corresponding noise density.}
	\label{fig5}
\end{figure}

	\begin{figure}[H]
		\centering
		\makebox[0pt]{
			\subfigure[]{\label{fig:gull}\includegraphics[width=0.25\textwidth]{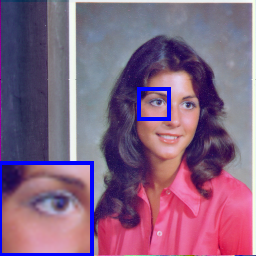}}
			\subfigure[]{\label{fig:tiger}\includegraphics[width=0.25\textwidth]{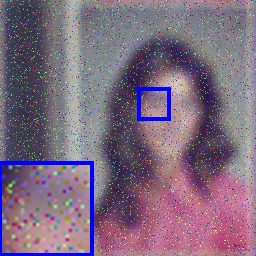}}
			\subfigure[]{\label{fig:tiger}\includegraphics[width=0.25\textwidth]{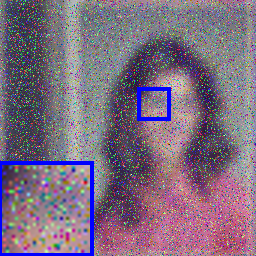}}
			\subfigure[]{\label{fig:tiger}\includegraphics[width=0.25\textwidth]{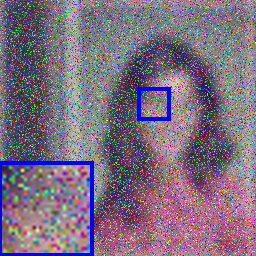}}
		}\\
		\vspace{-.2cm}
		\makebox[0pt]{
			\subfigure[]{\label{fig:gull}\includegraphics[width=0.25\textwidth]{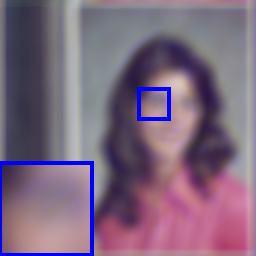}}
			\subfigure[]{\label{fig:tiger}\includegraphics[width=0.25\textwidth]{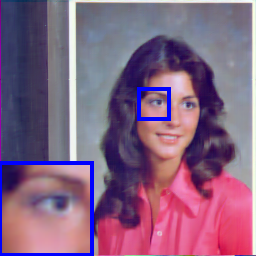}}
			\subfigure[]{\label{fig:tiger}\includegraphics[width=0.25\textwidth]{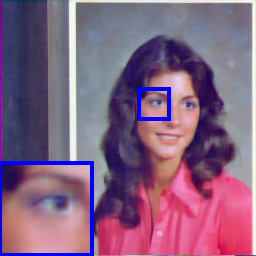}}
			\subfigure[]{\label{fig:tiger}\includegraphics[width=0.25\textwidth]{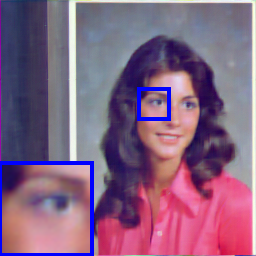}}
		}
		\vspace{-.3cm}
		\caption{(a) Clear image, (e) Blurred image, (b, c, d) Blurred and noisy images with noise densities of 
			$10\%, 
			20\%$ and 
			$30\%$, respectively; (f, g, h)  Restored images obtained by  the proposed method, each at the corresponding noise density.}
		\label{fig6}
	\end{figure}

	\begin{figure}[H]
		\centering
		\makebox[0pt]{
			\subfigure[]{\label{fig:gull}\includegraphics[width=0.5\textwidth]{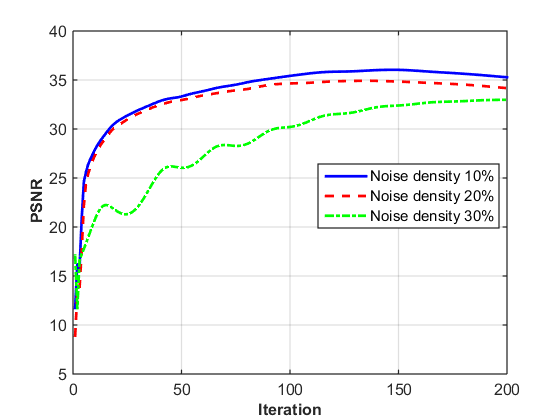}}
			\subfigure[]{\label{fig:tiger}\includegraphics[width=0.5\textwidth]{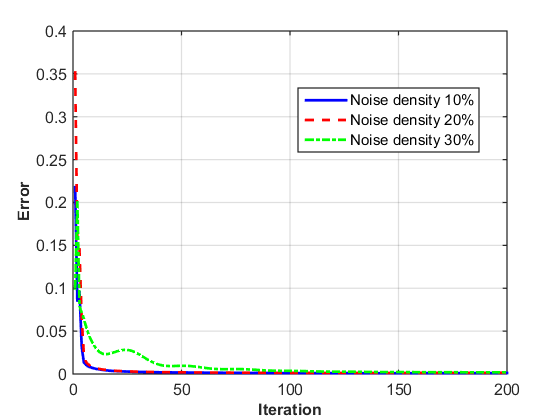}}
		}\\
		\vspace{-.3cm}
		\caption{Results for  different noisy density of the proposed method, (a) PSNR, (b) Relative error}
		\label{fig7}
	\end{figure}

	\section{Conclusion}

In this paper, we first presented a model based on the Minimax Concave Penalty for image denoising and deblurring. The resulting model has a nonconvex structure. We then proposed a numerical algorithm to solve this model using ADMM and extended the proposed method to color images. The convergence of the method is analyzed. In the section on numerical results, the quality of the retrieved images is evaluated using different types of tests, and the results show the efficiency of the method.


\begin{thebibliography}{99}
		\providecommand{\doi}[1]{DOI~\discretionary{}{}{}#1}
		\bibitem{nn01}
		Gonzalez, R.C. and Woods, R.E., 2000. \textit{Processamento de imagens digitais}. Editora Blucher.
		\bibitem{nna01}
		Sundararajan, D., 2026. Digital image processing: a signal processing and algorithmic approach. Springer Nature.
		\bibitem{nn02}
		Zhang, X., Bai, M. and Ng, M.K., 2017. Nonconvex-TV based image restoration with impulse noise removal. 
		SIAM Journal on Imaging Sciences, 10(3), pp.1627-1667.
		\bibitem{nn03}
		Adam, T., Paramesran, R., Mingming, Y. and Ratnavelu, K., 2021. Combined higher order non-convex total 
		variation with overlapping group sparsity for impulse noise removal. \textit{Multimedia Tools and 
			Applications}, 80(12), pp.18503-18530.
		\bibitem{nn04}
		Hansen, P.C., Nagy, J.G. and O'leary, D.P., 2006. \textit{Deblurring images: matrices, spectra, 
			and filtering}. Society for Industrial and Applied Mathematics.
		\bibitem{nn05}
		Rudin, L.I., Osher, S. and Fatemi, E., 1992. Nonlinear total variation based noise removal algorithms. 
		\textit{Physica D: nonlinear phenomena}, 60(1-4), pp.259-268.
		\bibitem{nn09}
		Naeem, M., Bhatti, S.M., Rashid, M., Jaffar, A., Akram, S., Fida, B. and Ahmad, A., 2026. 
		A robust deep learning approach for impulse noise filtering using hybrid auto-encoder with 
		fuzzy median filter. \textit{Plos one}, 21(4), p.e0343141.
		\bibitem{nn010}
        Yadav, N.K., Dhawan, A., Tiwari, M. and Jha, S.K., 2025. A state-of-the-art survey on noise 
        removal in a non-stationary signal using adaptive finite impulse response filtering: challenges, 
        techniques, and applications. \textit{International Journal of Systems Science}, 56(4), pp.885-918.
		\bibitem{nn011}
		Zhang, K., Ren, W., Luo, W., Lai, W.S., Stenger, B., Yang, M.H. and Li, H., 2022. Deep image 
		deblurring: A survey. \textit{International Journal of Computer Vision}, 130(9), pp.2103-2130.
		\bibitem{na1}
		Wen, C., Liang, W., Lu, J., Xu, C. and Zou, Y., 2026. Fractal image coding with deep image 
		prior and total variation regularization for salt-and-peppers noise removal. \textit{Fractals}, 34(01), p.2550095.
		\bibitem{na2}
		Zhao, J.L., Zhang, T.H., Fang, S., Gao, J.F., Wang, J.Y. and Gong, M.G., 2025. 
		Spatial-spectral texture--preserved total variation: A novel regularization for hyperspectral image denoising. 
		\textit{IEEE Transactions on Circuits and Systems for Video Technology}.
		\bibitem{nn1}
         Yuan, G. and Ghanem, B., 2017. $\ell _0 $ TV: A Sparse 
         Optimization Method for Impulse Noise Image Restoration. \textit{IEEE transactions on 
         	pattern analysis and machine intelligence}, 41(2), pp.352-364.
       \bibitem{nna1}
         Yin, M., Adam, T., Paramesran, R. and Fikree Hassan, M., 2025. 
         An $\ell_0$ total generalized variation for impulse noise removal. \textit{Multimedia Tools and 
         	Applications}, 84(27), pp.31813-31839.
		\bibitem{nn2}
		Yang, J., Zhang, Y. and Yin, W., 2009. An efficient TVL1 algorithm for deblurring multichannel 
		images corrupted by impulsive noise. \textit{SIAM Journal on Scientific Computing}, 31(4), pp.2842-2865.
		\bibitem{nn3}
		Yang, J., Zhang, Y. and Yin, W., 2010. A fast alternating direction method for TVL1-L2 signal 
		reconstruction from partial Fourier data. \textit{IEEE Journal of Selected Topics in Signal Processing}, 4(2), pp.288-297.
		\bibitem{nn4}
		Ge, D., Jiang, X. and Ye, Y., 2011. A note on the complexity of L p minimization. 
		\textit{Mathematical programming}, 129(2), pp.285-299.
		\bibitem{nn08}
		Zhang, C.H., 2010. Nearly unbiased variable selection under minimax concave penalty. 
				\textit{Annals of statistics}, 38(2), pp.894-942.
		\bibitem{nn06}
		Selesnick, I., 2017. Total variation denoising via the Moreau envelope. \textit{IEEE Signal 
			Processing Letters}, 24(2), pp.216-220.
		\bibitem{nn07}
		Storath, M., Weinmann, A. and Demaret, L., 2014. Jump-sparse and sparse recovery using Potts 
		functionals. \textit{IEEE Transactions on Signal Processing}, 62(14), pp.3654-3666.
		\bibitem{nn5}
		You, J., Jiao, Y., Lu, X. and Zeng, T., 2019. A nonconvex model with minimax concave penalty for image restoration. 
		\textit{Journal of Scientific Computing}, 78(2), pp.1063-1086.
		\bibitem{nn6}
		Gong, Y. and Zhao, P., 2026. Improved model based on weighted fractional-order overlapping group sparsity with
		 minimax concave penalty for image denoising and deblurring. \textit{Journal of Electronic Imaging}, 35(2), pp.023044-023044.
		\bibitem{nn7}
		Liu, X., Yu, X. and Liang, J., 2025. Minimax concave penalty-based TGV model for Poissonian image 
		restoration. \textit{Mathematics and Computers in Simulation}, 238, pp.25-44.
		\bibitem{nn8}
		Pokala, P.K., Hemadri, R.V. and Seelamantula, C.S., 2021. Iteratively reweighted 
		minimax-concave penalty minimization for accurate low-rank plus sparse matrix decomposition. 
		\textit{IEEE Transactions on Pattern Analysis and Machine Intelligence}, 44(12), pp.8992-9010.
		\bibitem{nn9}
		Yang, X.J., 2019. General fractional derivatives: theory, methods and applications. CRC press.
		\bibitem{nn10}
         He, N., Wang, J.B., Zhang, L.L. and Lu, K., 2015. An improved fractional-order differentiation model for image denoising.
         \textit{signal Processing}, 112, pp.180-188.
		\bibitem{nn11}
		Han, B., 2017. Framelets and wavelets. Algorithms, Analysis, and Applications. Birkh\"auser/Springer, Cham.
		\bibitem{nn12}
		Wang, Y., Yin, W. and Zeng, J., 2019. Global convergence of ADMM in nonconvex nonsmooth optimization. \textit{Journal of Scientific Computing}, 78(1), pp.29-63.
		\bibitem{nn14}
		Hong, M., Luo, Z.Q. and Razaviyayn, M., 2016. Convergence analysis of alternating direction method of multipliers for a family of nonconvex problems. \textit{SIAM Journal on Optimization}, 26(1), pp.337-364.
		
		
	   \bibitem{nn15}
		Dong, F. and Chen, Y., 2016. A FRACTIONAL-ORDER DERIVATIVE BASED VARIATIONAL FRAMEWORK FOR IMAGE DENOISING. \textit{Inverse Problems \& Imaging}, 10(1), p.27.
		\bibitem{nn16}
		Chen, D., Chen, Y. and Xue, D., 2013. Fractional‐order total variation image restoration based on primal‐dual algorithm. In \textit{Abstract and Applied Analysis} (Vol. 2013, No. 1, p. 585310). Hindawi Publishing Corporation.
		
		\bibitem{nn17}
		Lou, Y. and Yan, M., 2018. Fast L1-L2 minimization via a proximal operator. \textit{Journal of Scientific Computing,} 74(2), pp.767-785.
		\bibitem{nn18}
		You, J., Jiao, Y., Lu, X. and Zeng, T., 2019. A nonconvex model with minimax concave penalty for image restoration. \textit{Journal of Scientific Computing}, 78(2), pp.1063-1086.
		\bibitem{nn19}
		Lian, W. and Liu, X., 2023. Non-convex fractional-order TV model for impulse noise removal. J\textit{ournal of Computational and Applied Mathematics}, 417, p.114615.
		\bibitem{nn20}
		Liu, X., 2017. Alternating minimization method for image restoration corrupted by impulse noise. \textit{Multimedia Tools and Applications}, 76(10), pp.12505-12516.
		
		
		
	\end{thebibliography}
\end{document}